\documentclass[10pt]{article}
\usepackage[letterpaper]{geometry}
\usepackage{hicss}
\usepackage{times}
\usepackage[none]{hyphenat}
\usepackage{url}
\usepackage{latexsym}
\usepackage{listings}
\usepackage{indentfirst}
\usepackage{graphicx}
\graphicspath{{/}}
\usepackage[
    style=apa,
  ]{biblatex}
\usepackage{adjustbox}
\usepackage{comment}
\usepackage{amsthm}
\usepackage{subfig}
\newtheorem{definition}{Definition}
\newtheorem{theorem}{Theorem}
\newtheorem{lemma}[theorem]{Lemma}
\usepackage{amsmath,amssymb,amsfonts}
\usepackage[ruled,vlined]{algorithm2e}
\title{Methodologies for Selection of Optimal Sites for Renewable Energy Under a Diverse Set of Constraints and Objectives}

 \author{Arunabha Sen\textsuperscript{*}, Christopher Sumnicht\textsuperscript{*}, Sandipan Choudhuri\textsuperscript{*}, Suli Adeniye\textsuperscript{*}, Amit B. Sen\textsuperscript{$\dagger$}\\
  \textsuperscript{*}Arizona State University\\ \textsuperscript{$\dagger$}Palantir Technologies\\
  {{\{asen, csumnich, s.choudhuri, sadeniye\}@asu.edu, amit@rilke.us}} 
}

\date{}

\begin{document}
\maketitle
\begin{abstract}
In this paper, we present methodologies for optimal selection for renewable energy sites under a different set of constraints and objectives. We consider two different models for the site-selection problem - coarse-grained and fine-grained, and analyze them to find solutions. We consider multiple different ways to measure the benefits of setting up a site. We provide approximation algorithms with a guaranteed performance bound for two different benefit metrics with the coarse-grained model. For the fine-grained model, we provide a technique utilizing Integer Linear Program to find the optimal solution. We present the results of our extensive experimentation with synthetic data generated from sparsely available real data from solar farms in Arizona. 
\end{abstract}

\subsubsection*{Keywords:}

Optimal site selection, Renewable energy, Seasonal variation in supply and demand, Network Flow, Mathematical Programming

% Include up to five keywords that capture the main topics or themes of the paper. Separate each keyword with a comma and space.

\section{Introduction}

In alignment with the goal of reducing $CO_2$ emissions in the power sector, most utility companies in the USA are currently in the process of decommissioning fossil fuel-based energy generating stations.  Salt River Project (SRP), the utility company providing electricity to the Phoenix metropolitan area, has an aggressive plan to decommission a number of its units in the next few years. SRP has announced that it will decommission four of its existing plants by 2032, \textcite{srpGridWater2023}.
%It will (i) exit its ownership of coal plants in Colorado in the mid-to late-2020s, (ii) close The Four Corners Power Plant no later than 2031, (iii) announce that the Coronado Generating Station (CGS) is scheduled to close no later than 2032 and (iv) it anticipates closing its share of Springerville Generating Station in the years following the closure of Coronado  \textcite{srpGridWater2023}. 
While decommissioning of these units will certainly reduce generation, electricity demand in cities like Phoenix is growing significantly due to migration from other states to Arizona, making Maricopa County, where Phoenix is located, the fastest growing county in the country \textcite{srpNTP2023}.
% By the end of 2030, most of these stations are expected to be decommissioned, which in turn, reduces conventional energy generation. Despite this, the energy demand in several southwestern states is increasing every year due to significant migration from north-eastern and mid-western states to south-western states.  
Moreover, the large-scale introduction of electric vehicles in the next few years will significantly increase energy demand \textcite{srp2023}
Under the circumstances, the role of renewable energy in compensating for the decrease in fossil fuel-based energy generation becomes paramount.  Solar and wind power present themselves as practical renewable energy alternatives, particularly in southwestern states like Arizona. Because of an abundance of sunlight in Arizona, several privately owned solar farms have appeared in the Arizona landscape and are currently generating a moderate amount of energy. In 2023 and 2024, more than 1000 MW of additional renewable energy is expected to be available to SRP for distribution to its consumers, indicating future growth of solar facilities \textcite{srpNTP2023}.  SRP is currently very actively engaged in necessary infrastructure expansion taking into account various socioeconomic and sustainability impacts of such an expansion \textcite{srpISP2023}.

\begin{figure*}[tbh]
\centering
\subfloat[]{\label{fig:1}\includegraphics[width=0.30\textwidth]{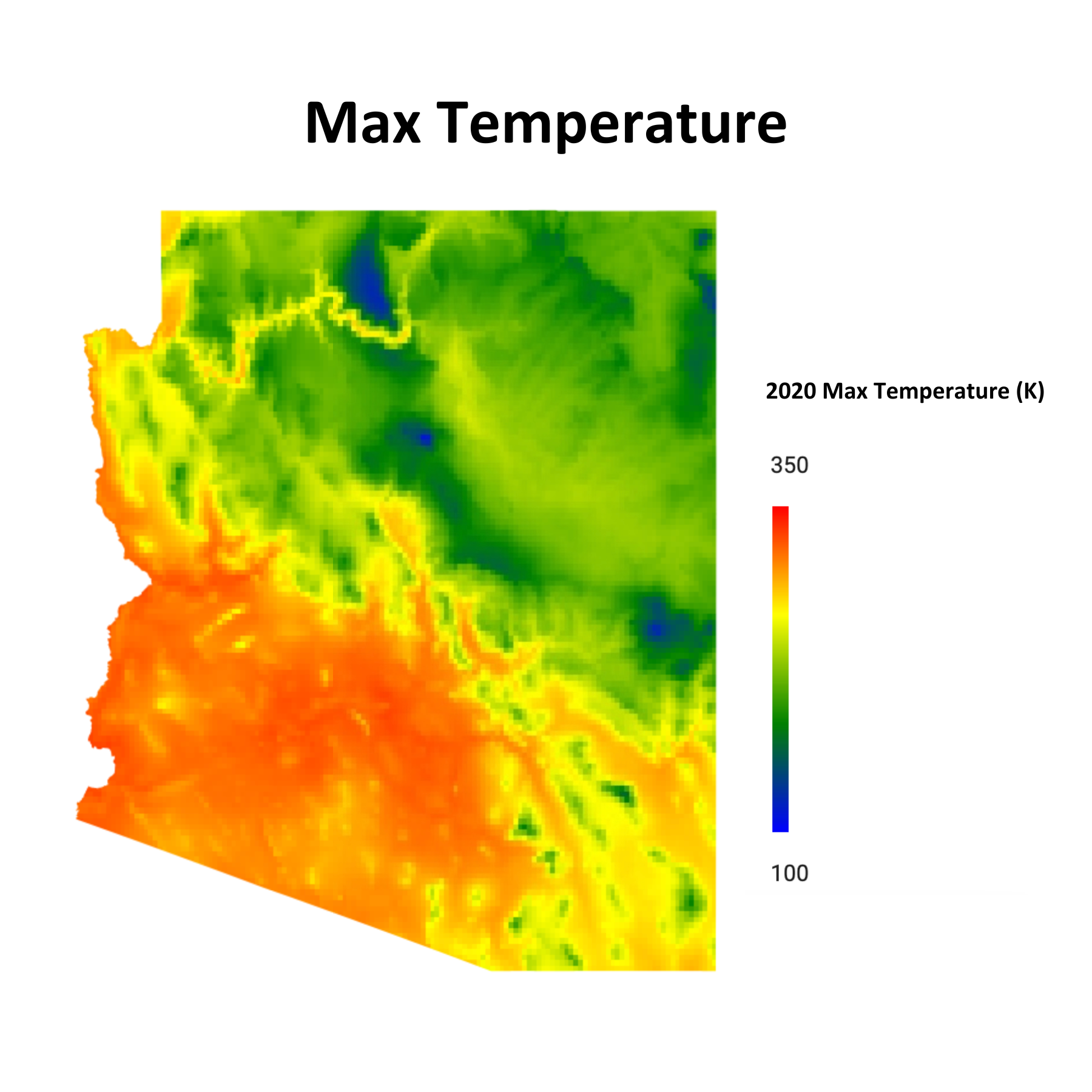}}\hfill
\subfloat[]{\label{fig:2}\includegraphics[width=0.30\textwidth]{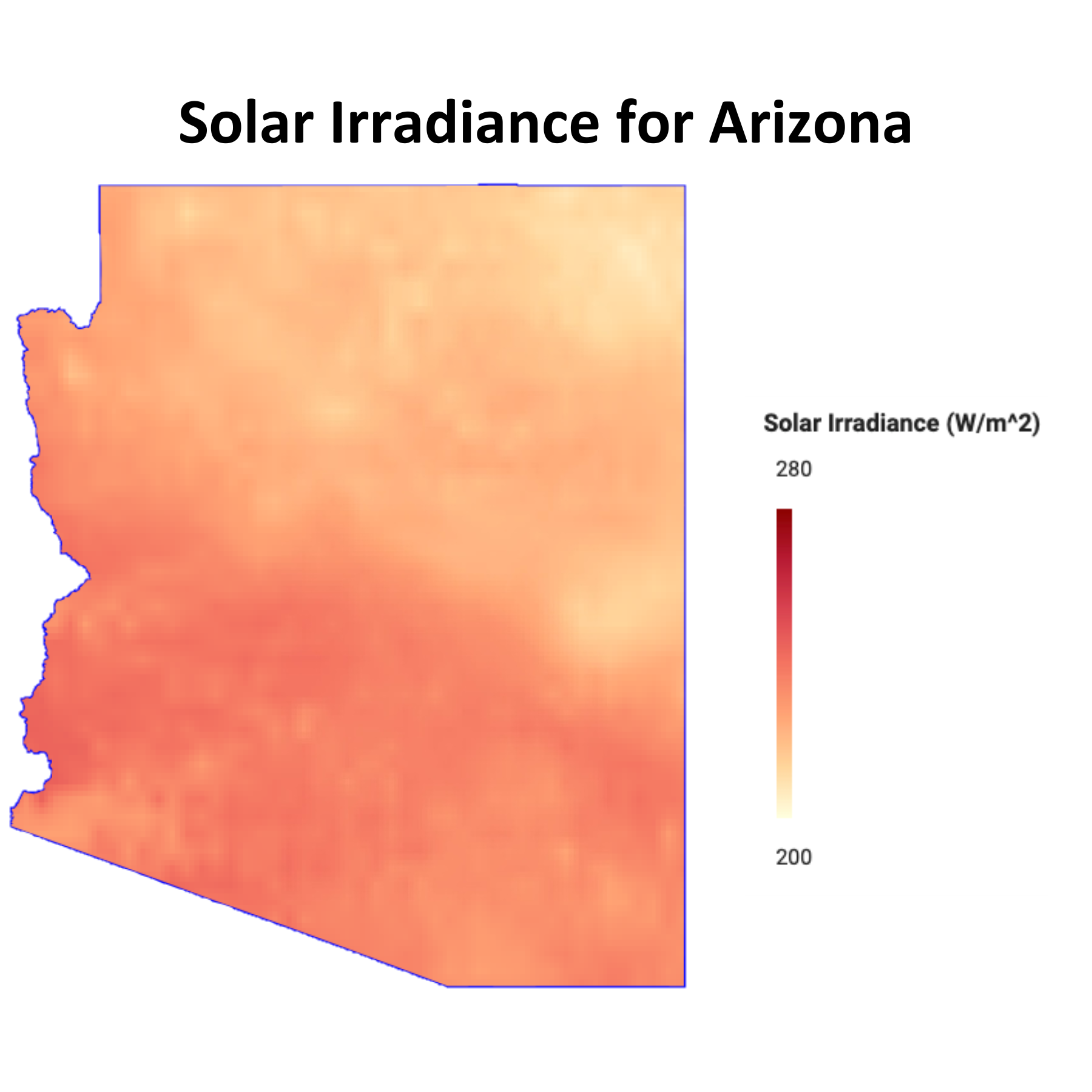}}\hfill
\subfloat[]{\label{fig:3}\includegraphics[width=0.30\textwidth]{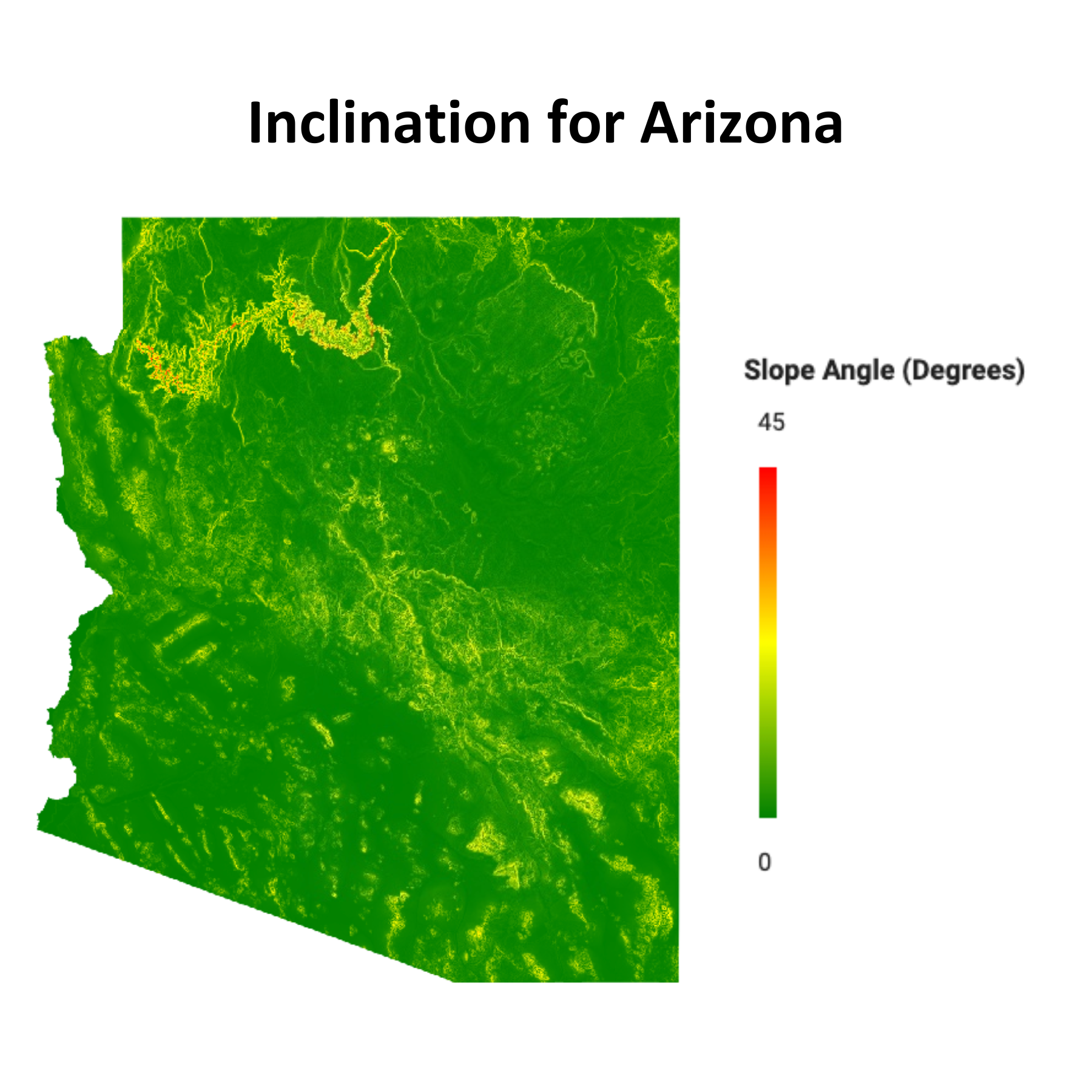}}
\caption{Arizona Climate and Orographic Characteristics}
\label{GIS}
\vspace{-5mm}
\end{figure*}

Identification of suitable sites for setting up solar/wind farms is critical, as poor site choices will have serious repercussions on meeting economic, societal and energy demands. Accordingly, the site selection process must take into account a multitude of factors, such as Climate: (Solar irradiation, Average temperature, Annual sunshine hours), (ii) Environment: (Land accessibility, Land use), (iii) Orography: (Land slope, Elevation, Land orientation), (iv) Location: (Distance to urban areas, transmission lines, substations),  (v) Economic: (Return on investment, Net profit, Land cost), (vi) Social: (Agricultural effect, Public interest, Public security), and (vii) Risk: (Environmental risk, Time-delay risk).

In the last few years, a large number of studies on optimal site selection via Multi-Criteria Decision Making (MCDM) have appeared in the scientific literature, 
\textcite{
turk2021multi, 
% kereush2017determining, 
% deveci2021evaluation, 
%bandira2022optimal, 
% akkas2017optimal, 
% mierzwiak2017multi, 
barzehkar2021decision, 
%akhtar2020site
}.
% koc2019multi}
% regarding optimal site selection via Multi-Criteria Decision Making (MCDM) process. 
Although different tools and techniques are available to solve the MCDM problems, the techniques of choice for researchers in this domain are Geographic Information System (GIS), \textcite{demers2009fundamentals, guaita2019analyzing} and Analytical Hierarchy Process (AHP), \textcite{saaty1994fundamentals, koc2019multi}. The GIS-based approaches take into account factors such as {\em Terrain Inclination}, {\em Solar Irradiance} and {\em Temperature}. Terrain Inclination is taken as a factor as the inclination or slope angle affects solar panels' efficiency by altering the angle at which sunlight strikes the panel surface. Generally, flatter terrains are preferable for solar power plant construction.  Although the AHP-based approaches are also extensively used,  it is well known that the AHP has limitations and may produce inconsistent results. A comprehensive evaluation of the AHP, based on an empirical investigation and objective testimonies by 101 researchers, found at least 30 flaws in the AHP and found it unsuitable for complex problems, and in certain situations even for small problems, \textcite{munier2021uses}.
% \begin{figure*}[tbh]
% \centering     %%% not \center
% \subfigure{\label{fig:1}\includegraphics[width=0.33\textwidth]{AZ_Temp_Distribution_Dia.png}}
% \subfigure{\label{fig:2}\includegraphics[width=0.33\textwidth]{AZ_Irradiance_Dia.png}}
% \subfigure{\label{fig:3}\includegraphics[width=0.33\textwidth]{AZ_Slope_Angle_Dia.png}}
% %% not \center
% \caption{Arizona Climate and Orographic Characteristics}
% \end{figure*}

Although both GIS and AHP-based approaches have limitations, they do capture certain aspects that are useful in the optimal site selection process. In this paper, we propose a Mathematical Programming+Network Flows-based technique to complement the GIS+AHP-based approach undertaken by many other researchers.  Our process comprises a two-phase scheme
where during the first phase, utilizing GIS and AHP, we conduct a preliminary site selection, and during the second phase, final site selection is conducted utilizing Network Flows and Mathematical Programming, \textcite{schrijver1998theory}.

In Phase I part of the analysis, for {\em Terrain Inclination} we used the United States Geological Survey National Elevation Dataset (NED) as the source for topographic data. To calculate slope angles across Arizona, Google Earth Engine's terrain algorithms were employed, resulting in a comprehensive slope dataset. For {\em Solar Irradiance} we used the University of Idaho Gridded Surface Meteorological Dataset (GRIDMET) to analyze solar irradiance across Arizona. We filtered this dataset for 2022 and selected the surface downward shortwave radiation $(W/m^2)$ variable. We generated an annual mean solar irradiance layer and clipped it to the Arizona boundaries. For {\em Temperature} we used the TerraClimate: Monthly Climate and Climatic Water Balance for Global Terrestrial Surfaces dataset filtered for 2020 and selected the maximum temperature (tmmx) variable. We generated an annual mean and clipped it to the Arizona boundaries. The computational results of {\em Terrain Inclination}, {\em Solar Irradiance}, {\em Temperature}of Arizona utilizing the datasets mentioned earlier are shown in Figure \ref{GIS}.

In our two-phase approach, we use currently existing techniques for Phase I and the novelty of our approach lies in the Phase II. In the following, we list our contributions: 
\begin{itemize}
\item We propose a Mathematical Programming and Network Flow based-technique to refine the site selection process after the preliminary site selection process is completed through GIS+AHP-based approaches in Phase I
\item We provide two different versions of optimal site election problem - {\em Coarse Grained} and {\em Fine Grained}, where the difference between the two is that the former doesn't take into account {\em transmission infrastructure cost and capacity}, while the later does
\item  For the Coarse Grained version of the problem, we provide three different metrics to measure the {\em benefit} (or the {\em utility}) of selected sites, namely {\em Interval Utility}, {\em Cumulative Sub-Interval Utility} and {\em Minimum Sub-Interval Utility}. These metrics provide options for the decision-makers to measure utility in the short or the long term.
\item We provide $(1 - 1/e)$ approximation algorithms for the computation of (i) Interval Utility and (ii) Cumulative Sub-Interval Utility
\item We show that the Minimum Sub-Interval Utility function  is not sub-modular and as such a simple greedy algorithm with a guaranteed performance bound may not exist for this problem, as it does for the Cumulative Sub-Interval Utility problem
\item For the Fine Grained version, we provide an Integer Linear Program to find the optimal sites
\item We provide results of our extensive evaluation to study the impact of the {\em infrastructure design budget}, {\em energy supply and demand} on the percentage of the {\em energy demand} that can be met with these additional renewable energy sites.  
\end{itemize}

The objective of our exercise is to find the locations of {\em optimal sites} subject to various sets of constraints, such as energy demand, design budget, etc. The difference between the Coarse-Grained and Fine-Grained versions of the problem is that while the former doesn't take into account the transmission line setup cost and the transmission line capacity, the latter does take these two factors into account.

\section{Related Work}

%\textcite{guaita2019analyzing,turk2021multi,kereush2017determining,deveci2021evaluation,bandira2022optimal,akkas2017optimal,mierzwiak2017multi,barzehkar2021decision,akhtar2020site,nguyen2022gis}

%\begin{comment}
%The works presented here center around the critical task of selecting suitable sites for solar power plants, employing Geographic Information System (GIS) and multi-criteria decision-making approaches. Effective site selection directly affects the efficiency and profitability of solar farms. By considering diverse criteria and leveraging GIS data, the authors aim to identify the optimal locations for solar power plants in various regions worldwide.

% The works presented here center around the critical task of selecting suitable sites for solar power plants, employing Geographic Information System (GIS) and multi-criteria decision-making approaches. Effective site selection is crucial as it directly affects the efficiency and profitability of solar farms. By considering diverse criteria and leveraging GIS data, the authors aim to identify the optimal locations for solar power plants in various regions across the globe.

%Due to the importance of the selection of renewable energy generation sites, a significant number of studies have been undertaken by various research groups. In the following, we summarize some of the results that are available in the published literature.

In recent studies, researchers have employed Geographic Information System (GIS) databases to analyze various criteria for the sustainable development of solar photovoltaic power. \textcite{guaita2019analyzing} utilize GIS to assess factors such as land availability, solar radiation potential, and proximity to infrastructure, enabling the identification of suitable areas for solar farm establishment. Similarly, \textcite{turk2021multi} focuses on the specific case of the Erzurum province in Turkey.

Another line of research concentrates on determining the key criteria for selecting optimal sites for solar farms. \textcite{kereush2017determining} propose a comprehensive methodology that combines GIS analysis and multi-criteria decision-making techniques to evaluate environmental, technical, and economic factors. \textcite{deveci2021evaluation} focus on evaluating criteria for solar photovoltaic projects using fuzzy logarithmic additive estimation of weight coefficients and GIS. %Their study considers factors such as solar radiation, land suitability, economic viability, and proximity to infrastructure to determine their importance in site selection.

GIS-based multi-criteria decision-making (MCDM) techniques have been employed in different regions to identify suitable locations for solar farms. \textcite{mierzwiak2017multi} perform a multi-criteria analysis for solar farm location suitability by developing a comprehensive framework that considers solar irradiation, land use, and environmental constraints to identify suitable locations for solar farms. %These studies optimize the grid integration of solar power plants by identifying the most suitable locations. 

%\textcite{nguyen2022gis} utilize GIS-based simulation for solar farm site selection in south-central Vietnam. Their approach integrates various factors such as solar radiation, land suitability, and environmental constraints. By simulating different scenarios and evaluating the suitability of potential sites, they provide valuable insights to stakeholders and decision-makers involved in solar farm development in the region. 

In addition to solar energy, \textcite{barzehkar2021decision} apply GIS and MCDA techniques to select wind and solar farm sites in Isfahan Province, Iran. %By evaluating wind speed, solar radiation, land availability, and infrastructure proximity, they develop decision support tools for informed decision-making in renewable energy project site selection.  
The work by \textcite{koc2019multi} focuses on the multi-criteria site selection problem for wind and solar energy projects in Igdir Province, Turkey. They propose an approach based on Geographic Information Systems (GIS) and the Analytic Hierarchy Process (AHP) to evaluate different sites.

\section{Problem Formulation}

% \section{Part II: Medication-Assisted Treatment Facility Location}
\label{ProbForm}
In this section, we discuss both the {\em Coarse-Grained} and {\em Fine-Grained} models  %different models 
for the optimal site selection problem
% - {\em Coarse-Grained} and {\em Fine-Grained} 
and analyze the models to provide solutions to the problem. %The parameters used in these two models, the notations used and their explanations are provided in Table \ref{table1}. 
It may be noted that the distinction between the two models is that the Fine-Grained model use the last two parameters $LCost_{i,j}$  and 
$LCap_{i,j}$, while the Course-Grained model doesn't. The goal of both the models is to satisfy the largest percentage of the demand created at the demand (or load) points, subject to various constraints, such as {\em generation capacity} of each site $SCap_i$, facility setup cost for each site $SCost_i$ and the design budget $B$. Both the models take into account the variation of generation at site $S_i$ and the variation of demand (load) at load point $L_j$ over the time interval $0$ to $T$ (for instance, a time interval may correspond to a year, comprising multiple sub-intervals which may correspond to a month).

\subsection{Coarse-Grained Model and Analysis}

In this model, we assume that it is economically viable for a renewable energy generation facility to supply energy only to nearby {\em demand (load) points}. More precisely, we assume that a facility $S_i$ can supply energy only to those demand points that are within a circle of radius $R$. Our problem setting is illustrated on a map of a generic geographic area in Figure \ref{CoarseModel} where the blue circular dots indicate the {\em potential} locations ($PL_1$ through $PL_{10}$) of renewable energy generation sites and the red triangular dots indicate the location of demand points ($I_1$ through $I_{16}$). The circle with center at $PL_9$ indicates its {\em service range}. The solution is to select sites that maximizes a fraction of the demand, subject to various constraints if energy generation facilities are built in those locations. In Figure \ref{CoarseModel}, the blue circles in square boxes (sites $SL_1$ through $SL_4$) indicate the sites selected for that problem instance.

{% Meeting Largest Demand Fraction subject to Infrastructure Budget Constraint//

\begin{figure}[t]
	\begin{center}
	\vspace{-4.0mm}
		\includegraphics[width = 0.48\textwidth, keepaspectratio]{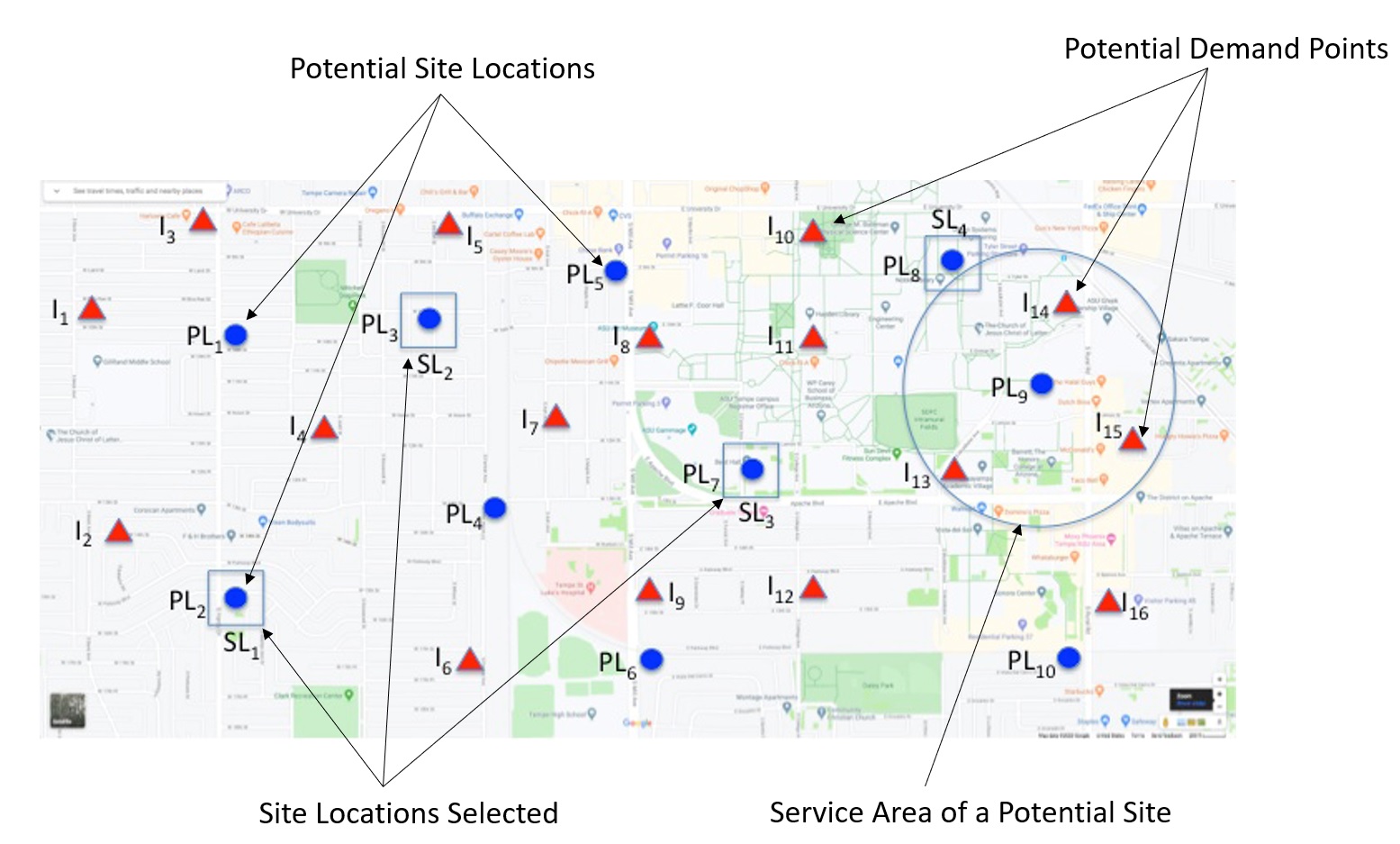}
	    \vspace{-4.0mm}
        \caption{Potential Generation Sites and Actual Demand Points}
        \vspace{-4.0mm}
		{\label{CoarseModel}}       
	\end{center}
 \vspace{-4.0mm}
\end{figure}

% \noindent
% Given,\\
% (i) A set of {\em opioid overdose incidences}, $I = \{I_1, \ldots, I_n\}$\\
% (ii) Each incidence $I_j, 1 \leq j \leq n$ has a {\em location} $l_j$, a {\em time} $t_j$, and a probability of occurrence $p_j$ associated with it. In this paper, we have studied the deterministic setting, where $p_j = 1 , \forall i 1 \leq i \leq n$. However, for a probabilistic analysis, one may just replace the counts of OOIs with the Expectation of OOIs. \\
% (iii) A set of {\em potential} MAT locations, $L = \{L_1, \ldots, L_m\}$\\
% (iv) MAT facility service area radius $R$\\
The goal of maximizing the fraction of the total demand is over a {\em time interval} $T$, which is divided evenly into $r$ sub-intervals $T_i, 1 \leq i \leq r$.  The time {\em interval} $0$ to $T$ is the observation period.  This implies $\sum_{i = 1}^r T_i = T$. 
\\
% (vi) Budget $B$\\

We define an {\em area of interest} $A$, where all demand (load) points, $l_i, 1 \leq i \leq n$, and potential renewable energy generation sites, $L_i, 1 \leq i \leq m$, are located. 
% locations $L_i, 1 \leq i \leq m$ are situated.
The service area associated with a facility at location $L_i$ will be denoted by $SA(L_i)$. In this Coarse-Grained Model, we make a few simplifying assumptions: (i)  the cost of building facilities at every location is {\em identical} and it is of {\em unit cost} ( $SCost_i = 1, \forall i, 1 \leq i \leq n$), (ii)  the demand at each demand point is either 0 or 1, (iii) the {\em budget} for building the infrastructure is an integer $B$. Accordingly, in this model, all subsets of the set of locations $L = \{L_1, \ldots, L_m\}$ of cardinality at most $B$ are potential solutions to the site selection problem.  It may be noted that an improper subset of the potential solutions will maximize the fraction of the total demand that can be met.
% We denote these subsets as  $\cal L$, i.e., ${\cal L} = \{ {\cal L}_1, {\cal L}_2, \ldots, {\cal L}_k\}$, where $ \forall j, 1 \leq j \leq k, {\cal L}_j \subseteq L$ and $|{\cal L}_j| \leq B$. 
The service area associated with a $L' \subseteq L$ will be denoted by $SA(L')$ and is defined as \(SA(L') = \cup_{L_i \in L} SA(L_i)\). The goal of the site selection problem is to identify the subset $L' \subseteq L$ that {\em maximizes} the ``benefit'', subject to the constraint that $|L'| \leq B$, where $B$ is the budget.   It may be noted that in this context, the ''benefit'' of setting up facilities at certain locations may be considered a fraction of the total demand that can be met by setting up facilities at those locations. % constraint that specifies the number of MAT facilities that can be set up.\\

% Our problem setting is illustrated on a map of a generic city in Figure \ref{MAT_facility_Example}. In this illustration, (i) the red triangles indicate the incidence locations, $I_1$ through $I_{16}$, (ii) the blue circles indicate the potential locations for setting up a MAT facility, $PL_1$ through $PL_{10}$, (iii) The circle around $PL_9$ indicates its service area, and (iv) blue circles in square boxes indicate the selected locations $SL_1, SL_2, SL_3$ that will maximize benefit, when budgetary constraints allow setting up of only 3 facilities (i.e., $B$ = 3).

We have defined the goal of the site selection problem as to identify locations for setting up renewable energy generation facilities (from among the set of potential locations), which maximizes ``benefit'' subject to budgetary constraints. Delving deeper, one has to address the question of how the ``benefit'' of setting up a generation facility should be measured over a period of time. In the following, we lay out three different metrics, any one of which can be used to measure the benefit of setting up a facility. We leave it up to the policymakers to decide as to which metric is most appropriate for their environment. In addition to laying out metrics, we also provide algorithms to find facilities that will maximize benefits. 
% for all these three metrics. 
The formal definition of these metrics is given below.

%\subsection{Deterministic Setting}

\begin{definition}
Interval Utility of a subset of a location set  $L' \subseteq L$, within the time period $T$, $IU(L', T)$), is defined as the {\em ratio} between the number of demand points with the demand  of one unit during time interval $T$ that can be met by the generating facilities located at $L'$, divided by the total demand  during time $T$, i.e.,   \[IU(L', T) = |I(SA(L'), T)|/|I(A, T)|\] where $|I(SA(L'), T)|$ represents the number of demand points in the service area of $L'$with demand of one unit during time interval $T$ and $|I(A, T)|$ is the total demand in the area of interest $A$ during time interval $T$.  
\end{definition}

The Interval Utility metric $IU(L', T)$ pays attention to the demand that can be met by $L'$ over the time interval $T$, without paying detailed attention to its effectiveness in each of time sub-intervals $T_i, 1 \leq i \leq r$. In order to capture the effectiveness of setting up facilities in locations $L'$ in each of time sub-intervals $T_i, 1 \leq i \leq r$, we define the following two metrics.

\begin{definition}
Cumulative Sub-Interval Utility of a subset of a location set  $L' \subseteq L$, over time interval $T$, where $T = \sum_{j = 1}^r T_i$,  denoted by $CSIU(L', T)$, is defined as the sum over $r$ time sub-intervals, the ratio between the demand that could have been served by the facilities at $L'$ during time sub-interval $T_i$, divided by the total number of incidents in $T_i$, i.e.,   \[CSIU(L', T) = \sum_{i = 1}^r (|I(SA(L'), T_i)|/|I (A, T_i)|)\] where $I(SA(L'), T_i)$ represents the demand during $T_i$ that could have been met if facilities existed in locations $L'$, and $|I(A, T_i)|$ is the total demand in the area of interest $A$ during time sub-interval $T_i$.
\end{definition}

Although the Cumulative Sub-Interval Utility metric $CSIU(L', T)$ pays attention to demand that can be met by $L'$ during each of time sub-intervals $T_i, 1 \leq i \leq r$, due to its {\em additive nature}, it may fail to recognize the scenario where the service in a particular time sub-interval $T_i$ is unacceptably poor. In order to remedy this shortcoming of  $CSIU(L', T)$, we define the following metric.

\begin{definition}
Minimum Sub-Interval Utility of subset of location set  $L' \subseteq L$ over time interval $T$, denoted by $MSIU(L', T)$, is defined as the minimum over $r$ time sub-intervals, the ratio between the demand that could have been met by the  facilities at $L'$ during time sub-interval $T_i$, divided by the total demand in $T_i$, i.e.,   \[MSIU(L', T) = \min_{i = 1}^r ~(|I(SA(L'), T_i)|/|I (A, T_i)|)\] where $I(SA(L'), T_i)$ represents the demand during $T_i$ that could have been met if facilities existed in a subset of locations $L'$, and $|I(A, T_i)|$ represents the total demand in the area of interest $A$ during time $T_i$.
\end{definition}

We provide intuition for the preceding metrics using an example:

\vspace{0.1in}
\noindent
{\bf Example:} Consider a problem instance where the budget constraint allows building only one facility. Suppose that $L'= \{l_a\}, L'' = \{l_b\}$, and $L''' = \{l_c\}$ are three different subsets of $L$, and they are being evaluated for their suitability as the ``best choice'' for the facility. Moreover, let $r = 2$ (i.e., only two time sub-intervals $T_1$, and $T_2$ need to be considered). The total demand and the demand met by different choices of locations $L', L''$, and $L'''$ are shown in Table \ref{ExampleTable}. 

% \begin{comment}
\begin{table}[tbh]
		\centering
		\setlength{\tabcolsep}{8pt}
        \renewcommand{\arraystretch}{1.3}
		\vspace{-4.0mm}
		\caption{Demand in time sub-intervals and the demand-met by different subsets of location}
  \begin{adjustbox}{width = 0.47\textwidth}
		\begin{tabular} {| c | c | c|}  \hline 
			     & Sub-interval $T_1$ & Sub-interval $T_2$\\ \hline
               Total Demand &  100 & 10\\ \hline
			    Demand met by $L'$ & 80 & 2\\ \hline
                Demand met by $L''$ & 70 & 4\\ \hline
			    Demand met by $L''$ & 55 & 5\\ \hline
		\end{tabular}
  \end{adjustbox}
        \label{ExampleTable}
        \vspace{-10.0pt}
	\end{table}
% \end{comment}
	
% 	\begin{table}[tbh]
% 		\centering
% 		\vspace{-4.0mm}
% 		\caption{Incidences in time sub-intervals and the number of them covered by different subsets of location}
% 		\begin{tabular} {| c | c | c| C|}  \hline 
% 			     & IU(*, T) & CSIU(*, T) & MSIU(*, T)\\ \hline
%               $* = L'$ & (80 + 2)/(100 + 10) =   & 80/100 + 2/10 =  & min (8/100, 2/10) = \\ \hline
% 			    $* = L''$  &  (70 + 4)/(100 + 10) =   & 70/100 + 4/10 =  & min (70/100, 4/10) = \\ \hline
%                 $* = L''$  &  (55 + 5)/(100 + 10) =   & 55/100 + 5/10 =  & min (55/100, 5/10) = \\ \hline
% 		\end{tabular}
%         \label{ExampleTable}
%         \vspace{-10.0pt}
% 	\end{table}

Whether facilities at $L'$, $L''$, or $L'''$ are built depends on the metric used. In this example, $T_1$ has a demand of $100$ and $T_2$ has a demand of $10$. Therefore, the total demand $T_1 + T_2 = 100 + 10 = 110$. It can be verified that under the Interval Utility metric, building the facility at $L'$ is optimal as $L'$ contributes a utility of $80 + 2 = 82$, which is greater than $L''$ contribution of $70 + 4 = 74$, and $L'''$ of $55 + 5 = 60$. However, if one examines that choice carefully, one finds shortcomings with that solution. The $L'$ solution covers 80\% of the incidences in $T_1$ and only 20\% in $T_2$. This is because the metric pays attention to the entire interval $T_1$ + $T_2$ but does not take in account to individual sub-intervals. If one wants to examine the quality of the solution with respect to each sub-interval $T_1$ and $T_2$, the $L''$ and $L'''$ solutions may be better than the $L'$ solution. If, instead, one uses the Cumulative Sub-Interval Utility, then $L''$ is the best choice as (70/100 + 4/10 = 1.1) is greater than (80/100 + 2/10 = 1.0) for $L''$ and (55/100 + 5/10 = 1.05) for $L'''$.  Similarly, if one wants to use Minimum Sub-Interval Utility, then $L'''$ is the best choice as a minimum (55/100, 5/10)  = 0.5) is greater than the minimum (80/100, 2/10)  = 0.2) for $L'$ and minimum (70/100, 4/10)  = 0.4) for $L''$. This discussion is summarized in Table \ref{UtilityResults}.

\begin{table}[tbh]
		\centering
		\setlength{\tabcolsep}{8pt}
        \renewcommand{\arraystretch}{1.3}
		\vspace{-4.0mm}
		\caption{Demand in time sub-intervals and the number of them met by different subsets of location}
		\begin{adjustbox}{width = 0.47\textwidth}
		\begin{tabular} {| c | c | c| c|}  \hline 
			      & IU($\cdot$, T) & CSIU($\cdot$, T) & MSIU($\cdot$, T)\\ \hline
               $L'$ & (80 + 2)/(100 + 10) = {\bf .75}  & 80/100 + 2/10 = 1 & min (80/100, 2/10) = 0.2\\ \hline
			    $L''$  &  (70 + 4)/(100 + 10) =  .67 & 70/100 + 4/10 = {\bf 1.1} & min (70/100, 4/10) = 0.4\\ \hline
                $L'''$  &  (55 + 5)/(100 + 10) =   .55 & 55/100 + 5/10 =  1.05 & min (55/100, 5/10) = {\bf 0.5}\\ \hline
		\end{tabular}
		\end{adjustbox}
        \label{UtilityResults}
        \vspace{-10.0pt}
\end{table}

\noindent
% Given,\\
% (i) A set of {\em incidences}, $I = \{I_1, \ldots, I_n\}$\\
% (ii) Each incidence $i_j, 1 \leq j \leq n$ has a {\em location} $l_j$, a {\em time} $t_j$, and a probability of occurrence $p_j$ associated with it. In the deterministic setting $p_j = 1 , \forall i. 1 \leq i \leq n$\\
% (iii) A set of {\em potential} MAT locations, $L = \{L_1, \ldots, L_m\}$\\
% (iv) A set of subsets, $\cal L$ of $L$, i.e., ${\cal L} = \{ {\cal L}_1, {\cal L}_2, \ldots, {\cal L}_k\}$, where $ \forall j, 1 \leq j \leq k, {\cal L}_j \subseteq L$ and $|{\cal L}_j| \leq B$\\
% %{\em potential MAT locations}, $L = \{L_1, \ldots, L_m\}$\\
% (v) The time {\em interval} $0$ to $T$ is the interval during which all the incidences have taken place, i.e., $ 0 \leq t_i \leq T, \forall t_i, 1 \leq  i \leq n$.  The time interval $T$ is divided into $r$ equal sized {\em sub-intervals} $T_i, 1 \leq i \leq r$. This implies $\sum_{j = 1}^r T_i = T$.\\
% (vi) Budget $B$

\begin{definition}
    Optimal Interval Utility (OIU) Solution:
{\em OIU} solution is defined as a subset of a location set  $L' \subseteq L$, that maximizes {\em Interval Utility} $IU(L', T)$, i.e., 
%denoted by $IU(L')$, is defined as the ratio between the number of incidences that could have been served in case there were MAT facilities at $L'$, divided by the total number of incidents, i.e.,   
\[OIU\_Solution~=~max_{\{L' \subseteq L, |L'| = B\}} IU(L', T) \] where $IU(L', T)$ is the Interval Utility defined earlier, and B is the budget.
% =  max_{\{L' \subseteq L, |L'| = k\}} |I(L')|/|I\]
\end{definition}

%where $I(L')$ represents the set of incidences that could have been {\em covered} or {\em serviced} if MAT facilities existed in locations $L'$.

In the following we first state a well studied combinatorial optimization problem, {\em Maximum Set Cover (MSC)} and then show that the OIU is equivalent to MSC. MSC is NP-complete and a $(1 - 1/e)$ factor approximation algorithm is known for it \textcite{vazirani2001approximation}.

\begin{definition}
Maximum Set Cover (MSC) Problem: Given a set $S = \{s_1, \ldots, s_n\}$ and subsets 
${\cal S} = \{S_1, ..., S_m\}$ ($S_i \subseteq S, 1 \leq i \leq m$) and an integer $B$, find the largest subset $S' \subseteq S$ that can be covered by using a subset ${\cal S}' \subseteq {\cal S}$, where $|{\cal S}'| \leq B$.  
\end{definition}

% The MSC Problem is well studied in literature and $(1 - 1/e)$ factor approximation algorithm is known for the problem \textcite{vazirani2001approximation}. We show that the Optimal Interval Utility problem is equivalent to the MSC problem and as such the approximation algorithm for the MSC problem can be utilized to solve the Optimal Interval Utility problem. 

Let $\mathcal{O}$ be an instance of OIU with time interval $T$ and budget $B$. Denote the locations of the demand points with demand equal to $1$ within the time interval $T$ as $D$, and the potential locations of renewable energy sites as $L = \{L_j : 1 \leq j \leq m\}$.
Each $L_j \in L$ has a covering area (i.e., the demand at these points can be met by the facility located at $L_j$), which in this model is assumed to be a circular area with radius $R$. Each location $L_j \in L$ will have an associated subset $D_j$ of the demand points $D$ (i.e., subset of the demand points whose demands can be met by the facility at $L_j$).

%In the OIU problem instance, the locations of the demand points with demand equal to 1 within the time interval $T$, (say, the set $D$) and the potential locations of renewable energy sites (set $L$) are given as the input.
%Each element $L_j, 1 \leq i \leq m$ of the location set $L$ has a covering area (i.e., the demand at these demand points can be met by the facility located at $L_j$) points, which in this model is assumed to be a circular area with radius $R$. 
%With this definition of a covering area, each location $L_j, 1 \leq i \leq m$ of the location set $L$ will have an subset $D_j$ of the demand points $D$ associated with it, i.e., $D_J$ is the subset of the demand points whose demands can be met by the facility at $L_j$. In the OIU problem will also have a constraint that determines the number of facilities that can be built, say $B$.

We create an equivalent instance $\mathcal{M}$ with set $S$, $\mathcal{S}$, and $B$ of MSC as defined above from $\mathcal{O}$ in the following way: \\
(i) Set $S := D$.\\
(ii) Set ${\cal S} := {\cal D}$.\\
(iii) Let the $B$ of $\mathcal{M}$ be the same as the $B$ for $\mathcal{O}$.

It is easy to see that $\mathcal{O} \equiv \mathcal{M}$. Hence, the $(1 - 1/e)$ approximation algorithm that exists for the MSC can be used to solve the OIU problem.

% We can now view the OIU problem as a MSC problem in the following way. \\
%(i) The set $D$ of OIU can be viewed as the set $S$ of MSC.\\
%(ii) The set ${\cal D} = \{D_1, \ldots, D_m\}$ of OIU can be viewed as the set ${\cal S} %= \{S_1, ..., S_m\}$ ($S_i \subseteq S, 1 \leq i \leq m$) of the MSC.\\
%(iii) The parameter $B$ of OIU can be set equal to the parameter $B$ of the MSC.

%With such a formulation, the OIU can be viewed as a MSC problem and as such the $(1 - 1/e)$ approximation algorithm that exists for the MSC can be used to solve the OIU problem.

% \subsubsection{Cumulative Sub-Interval Utility Solution}
\vspace{0.05in}

Next, we focus our attention on the Cumulative Sub-Interval Utility function and show that it is submodular. For ease of reference, we define a {\em submodular} function next.

\begin{definition}
Submodular function: If $\Omega$ is a finite set, a submodular function is a set function, $f: 2^{\Omega} \rightarrow \mathbb{R}$, where 
$2^{\Omega}$ denotes the power set of $\Omega$ which satisfies the following condition:
For every $X, Y \subseteq \Omega$ with $X \subseteq Y$ and every $z \in \Omega \setminus Y$, we have 
$f(X \cup \{z\}) - f(X) \geq f(Y \cup \{z\}) - f(Y)$. (\textcite{nemhauser1978analysis}).
\end{definition}

\begin{lemma}
Cumulative Sub-Interval Utility function of a subset of a location set  $L' \subseteq L$ over  time interval $T$ (see Definition 2) is {\em monotonic}.
% denoted by $CSIU(L', T)$,
% defined as 
% the {\em sum} over $p$ time sub-intervals, the ratio between the number of incidences that could have been served in case there were MAT facilities at $L'$ during time sub-interval $T_i$, divided by the total number of incidents in $T_i$, i.e.,   
% \(CSIU({\cal L}_j) = \sum_{i = 1}^r (|I(L', T_i)|/|I (T_i)|)\) is monotonic.
% where $I({\cal L}_j, T_i)$ represents the set of incidences during $T_i$ that could have been {\em covered} or {\em serviced} if MAT facilities existed in locations $L'$, and $I(T_i)$ represents the total number of incidences in $T_i$.
\label{label1}
\end{lemma}

% \begin{proof}
% In order to prove this lemma, we need to show that for a pair ${\cal L}_X, {\cal L}_Y \subseteq {\cal L}$ ${\cal L}_X \subseteq {\cal L}_Y$ and $L_z \in L}$, \(CSIU({\cal L}_X \cup \{L_z\}) - CSIU({\cal L}_X) \geq CSIU({\cal L}_Y \cup \{L_z\}) - CSIU({\cal L}_Y)\).
% \end{proof}

\begin{lemma}
Cumulative Sub-Interval Utility function, 
%of subset of a location set  $L' \subseteq L$, 
% denoted by 
%defined as 
% the {\em sum} over $p$ time sub-intervals, the ratio between the number of incidences that could have been served in case there were MAT facilities at $L'$ during time sub-interval $T_i$, divided by the total number of incidents in $T_i$, i.e.,   
\(CSIU({L', T})\) 
%\(CSIU(L') = \sum_{i = 1}^r (|I(L', T_i)|/|I (T_i)|)\) 
is submodular.
% where $I(L', T_i)$ represents the set of incidences during $T_i$ that could have been {\em covered} or {\em serviced} if MAT facilities existed in locations $L'$, and $I(T_i)$ represents the total number of incidences in $T_i$.
\label{label2}
\end{lemma}
%\vspace{2in}
\begin{proof}
In order to claim that \(CSIU({L})\) is submodular, we need to show that for any pair $L', L'' \in L$,   $L' \subseteq L''$ and $L_i, \in {L \setminus L''}, CSIU(L' \cup \{L_i\}) - CSIU(L') \geq  CSIU(L'' \cup \{L_i\}) - CSIU(L'')$.

\vspace{0.1in}
\noindent
$CSIU(L' \cup \{L_i\}) - CSIU(L')$\\
$= \sum_{k = 1}^r \frac{|I(SA(L') \cup \{L_i\}, T_k)|}{|I(A, T_k)|} - \sum_{k = 1}^r \frac{|I(SA(L'), T_k)|}{|I(A, T_k)|}$ \\
$= \sum_{k = 1}^r \frac{|I(SA(L'), T_k)| + |I(SA(\{L_i\}), T_k)| - |I(SA(L' \cap \{L_i\}), T_k)|}{|I(A, T_k)|}$\\
$- \sum_{k = 1}^r \frac{|I(SA(L'), T_k)|}{|I(A, T_k)|}$ \\
$= \sum_{k = 1}^r \frac{|I(SA(L'), T_k)|}{|I(A, T_k|)} +  \sum_{k = 1}^r \frac{|I(SA(\{L_i\}), T_k)|}{|I(A, T_k)|}$\\ 
$-\sum_{k = 1}^r \frac{|I(SA(L') \cap \{L_i\}, T_k)|}{|I(A, T_k)|} - \sum_{k = 1}^r \frac{|I(SA(L'), T_k)|}{|I(A, T_k)|}$\\
$= \sum_{k = 1}^r \frac{|I(SA(L_i), T_k)|}{|I(A, T_k|)} - \sum_{k = 1}^r \frac{|I(SA(L') \cap \{L_i\}, T_k)|}{|I(A, T_k)|}$
$\geq \sum_{k = 1}^r \frac{|I(SA(L_i), T_k)|}{|I(A, T_k)|} - \sum_{k = 1}^r \frac{|I(SA(L'') \cap \{L_i\}, T_k)|}{|I(A, T_k)|}$\\
(as $L' \subseteq L''$, $|I(SA(L'), T_k)| \leq |I(SA(L'' \cap \{L_i\}), T_k)|$)\\
$= \sum_{k = 1}^r \frac{|I(SA(L''), T_k)|}{|I(A, T_k|)} +  \sum_{k = 1}^r \frac{|I(SA(\{L_i\}), T_k)|}{|I(A, T_k)|}$\\ 
$-\sum_{k = 1}^r \frac{|I(SA(L'') \cap \{L_i\}, T_k)|}{|I(A, T_k)|} - \sum_{k = 1}^r \frac{|I(SA(L''), T_k)|}{|I(A, T_k)|}$\\
$= \sum_{k = 1}^r \frac{|I(SA(L''), T_k)| + |I(SA(\{L_i\}), T_k)| - |I(SA(L'' \cap \{L_i\}), T_k)|}{|I(A, T_k)|}$\\
$- \sum_{k = 1}^r \frac{|I(SA(L''), T_k)|}{|I(A, T_k)|}$ \\
$= \sum_{k = 1}^r \frac{|I(SA(L'') \cup \{L_i\}, T_k)|}{|I(A, T_k)|} - \sum_{k = 1}^r \frac{|I(SA(L''), T_k)|}{|I(A, T_k)|}$ \\
$= CSIU(L'' \cup \{L_i\}) - CSIU(L'')$
\end{proof}

%\begin{array}{ll}
%      &   CSIU({\cal L}_X \cup \{L_z\}) - CSIU({\cal L}_X) \\
%       =    & $\sum_{k = 1}^r \frac{|I({\cal L}_X \cup \{L_z\}, T_k)|}{|I(T_k|}$ - $\sum_{k = 1}^r \frac{|I({\cal L}_X, T_k)|}{|I(T_k|}$}
 %  =    & {\sum_{k = 1}^r \frac{|I({\cal L}_X \cup \{L_z\}), T_k)|}{|I(T_k)|}} - {\sum_{k = 1}^r \frac{|I({\cal L}_X), T_k)|}{|I({T_k})|}}
%\end{array}

%\vspace{0.1in}
It has been shown in \textcite{nemhauser1978analysis} that a greedy heuristic for the problem of maximizing a monotone submodular function subject to a cardinality constraint admits a $ (1 - 1/e)$ approximation algorithm.
\begin{theorem}
A greedy heuristic for the Cumulative Sub-Interval Utility function $(CSIU(L')$ computation problem provides a $(1 - 1/e)$ approximation guarantee. 
\end{theorem}

\begin{proof}
Follows from lemmas \ref{label1}, \ref{label2} and \textcite{nemhauser1978analysis}.
\end{proof}

\vspace{0.1in}
In the following we show that unlike the {\em Cumulative Sub-Interval Utility function}  $CSIU(L',T)$, the {\em Minimum Sub-Interval Utility function} $MSIU(L',T)$ {\em is not submodular}.
% \subsubsection{Minimum Sub-Interval Utility Solution}

\begin{figure}[tbh]
	\begin{center}
	\vspace{-4.0mm}
		\includegraphics[width = 0.5\textwidth, keepaspectratio]{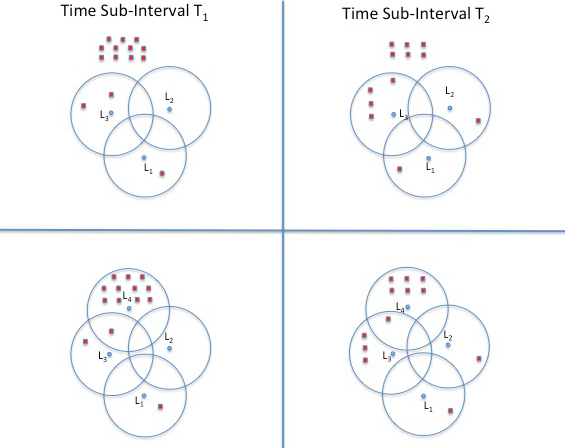}
	    \vspace{-4.0mm}
        \caption{Problem instance to demonstrate that $MSIU(L', T)$ is not submodular}
        \vspace{-4.0mm}
		{\label{MSIU_Example}}       
	\end{center}
\end{figure}

\begin{theorem}
Minimum Sub-Interval Utility function of a subset of a location set  $L' \subseteq L$, 
% denoted by $CSIU(L')$,
defined as 
% the {\em sum} over $p$ time sub-intervals, the ratio between the number of incidences that could have been served in case there were MAT facilities at $L'$ during time sub-interval $T_i$, divided by the total number of incidents in $T_i$, i.e.,   
\(MSIU(L', T) = \min_{i = 1}^r ~(|I(SA(L'), T_i)|/|I (A, T_i)|)\)
%\(MSIU(L') = \min_{i = 1}^r ~(|I(L', T_i)|/|I (T_i)|)\) 
is not submodular.
% where $I(L', T_i)$ represents the set of incidences during $T_i$ that could have been {\em covered} or {\em serviced} if MAT facilities existed in locations $L'$, and $I(T_i)$ represents the total number of incidences in $T_i$.
\end{theorem}

\begin{proof}

If \(MSIU({L, T})\) is submodular, then for any pair $L', L'' \in L$,   $L' \subseteq L''$ and $L_i, \in {L \setminus L''}$, \(MSIU(L' \cup \{L_i\}), T) - MSIU(L', T) \geq  MSIU((L'' \cup \{L_i), T\}) - MSIU(L'', T)\). In Figure \ref{MSIU_Example}, we provide a counterexample that demonstrates that that there exists
$L', L'' \in L$,   $L' \subseteq L''$ and $L_i, \in {L \setminus L''}$ such that $MSIU((L' \cup \{L_i), T\}) - MSIU(L', T) <  MSIU(L'' \cup \{L_i\}) - MSIU(L'')$.

Consider the instance shown in the Figure \ref{MSIU_Example}, where $L' = \{L_1, L_2\}$, $L'' = \{L_1, L_2, L4\}$ and $L_i = \{L_3\}$. The red squares indicate the locations of demand points . It may be seen from the diagram that 14 and 12 units of demand existed during the time sub-intervals $T_1$ and $T_2$ respectively. From the Figure \ref{MSIU_Example}, the following observations can be made: (i) the facilities in $L'$ can meet demand of 1 out of 14 units in $T_1$ and 2 out of 12 units during $T_2$, (ii) the facilities in $L'$ and $L_i$ can cover 3 out of 14 units of demand in $T_1$ and 6 out of 12 units of demand during $T_2$, (i) the facilities in  $L'''$ can cover 13 out of 14 units of demand in $T_1$ and 9 out of 12 units of demand during $T_2$, (i) the facilities in $L''$ and $L_i$ can cover 14 out of 14 units of demand in $T_1$ and 12 out of 12 units of demand during $T_2$, The results are tabulated in Table \ref{MSIUResults}. 

\begin{table}[tbh]
		\centering
		\setlength{\tabcolsep}{8pt}
        \renewcommand{\arraystretch}{1.3}
		\vspace{-2.0mm}
		\begin{adjustbox}{width = 0.47\textwidth}
		\begin{tabular} {| c | c | c| c|}  \hline 
			      & MSIU($\cdot$, $T_1$) & MSIU($\cdot$, $T_1$) & MSIU($\cdot$, $T_1 + T_2$)\\ \hline
               $SA(L')$ & 1/14 = .07  & 2/12 = .17 & min (.07, .17) = .07\\ \hline
			   $SA(L') \cup \{SA(L_i)\}$ & 3/14 = .21  & 6/12 = 0.5 & min (.21, 0.5) = ,21\\ \hline
			    $SA(L'')$ & 13/14 = .93  & 9/12 = .75& min (.93, .75) = .75\\ \hline
			   $SA(L'') \cup \{SA(L_i)\}$ & 14/14 = 1  & 12/12 = 1 & min (1, 1) = 1\\ \hline
		\end{tabular}
		\end{adjustbox}
        \label{MSIUResults}
        \caption{Incidences in time sub-intervals and the number of them covered by different subsets of location}
        \vspace{-10.0pt}
\end{table}

From the results in the Table \ref{MSIUResults}, it can be verified that $MSIU(L' \cup \{L_i\}) - MSIU(L') = .21 - .07 <  1 - .75 = MSIU(L'' \cup \{L_i\}) - MSIU(L'')$. Accordingly, this counterexample establishes that the $MSIU(\cdot, \cdot)$ is not a sub-modular function.
\end{proof}

% \vspace{1in}

\subsection{Fine-Grained Model and Analysis}

As discussed earlier, in addition to all the parameters used in the Coarsed-Grain Model, the Fine-Grained model takes into account two additional parameters (i) cost $LCost_{i,j}$ of building a transmission line connecting generating site $S_i$ to (demand) load point $L_j$, and (ii) maximum capacity of the transmission line $LCap_{i,j}$ connecting site $S_i$ to demand point $L_j$. Moreover, (i) energy demand at each of the demand (load) points is no longer required to be 0 or 1, and can take on any arbitrary demand values, (ii) the cost $SCost_i$ of building a site at $S_i$ is no longer required to be identical and of one unit, but allowed to take on any arbitrary value, (iii)  Variation of generation at site $S_i$ over time interval 0 to $T$ can take on arbitrary values, and 
(iv) variation of demand (load) at load point $L_j$ over the time interval 0 to $T$ can also take on any arbitrary values. The variables used in the Fine-Grained analysis and their explanations are summarized in Table \ref{table1}.

\begin{table}[tbh]
 \centering
 \begin{tabular}{ |c||c|} 
 \hline
Notation & Explanation \\ \hline \hline
$\{S_i\},$ & A set of {\em potential} renewable \\ 
$1 \leq i \leq n$ & energy generation sites \\ \hline
$SCost_i$  & Cost of building a site at $S_i$\\ \hline
$SCap_i$ & Maximum capacity of \\
 & generation at site $S_i$ \\ \hline
%$VG(S_i, 0-T)$ & Variation of generation at site $S_i$\\
%& over the time interval $0$ to $T$ \\
%& (for instance, each time interval\\
%& may correspond to a year)\\ \hline
%$VL(L_j, 0-T)$ & Variation of demand (load) at load \\
%& point $L_j$ over the time interval $0$ to $T$ \\ \hline
$B$ & Budget allocated for the design\\
& of renewable energy infrastructure  \\ \hline
$LCost_{i,j}$& Cost of building a transmission line\\
& connecting generating site $S_i$ to\\
& (demand) load point $L_j$ \\ \hline
$LCap_{i,j}$ & Maximum capacity of the \\
& transmission line connecting site $S_i$\\
& to demand point $L_j$\\ \hline
% & connecting generating site $S_i$ to\\
% & (demand) load point $L_j$ \\ \hline
% $\{L_j\}, 1 \leq i \leq m$ & A set of demand (load) points \\ \hline
% $\{L_j\}, 1 \leq i \leq m$ & A set of demand (load) points \\ \hline
% $\{L_j\}, 1 \leq i \leq m$ & A set of demand (load) points \\ \hline
% $\{L_j\}, 1 \leq i \leq m$ & A set of demand (load) points \\ \hline
% $\{L_j\}, 1 \leq i \leq m$ & A set of demand (load) points \\ \hline

 %  $p_1$ & $d_{1,1}$ & .  & $d_{i,j}$ & . & $d_{1,n}$ & MaxD(1) & MinD(1)\\ \hline
 %  $p_2$ & $d_{2,1}$ & .  & $d_{2,j}$ & . & $d_{2,n}$ & MaxD(2) & MinD(2) \\ \hline
 % . & . & . & . &  . & . & . &  .\\ \hline
 %  . & . & . & . &  . & . & . &  . \\ \hline
 %  $p_i$ & $d_{1,1}$ & . &  $d_{i,j}$ & . & $d_{i,n}$ & MaxD(i) & MinD(i) \\ \hline
 %  . & . & . & . &  . & . & . &  . \\ \hline
 %   . & . & . & . &  . & . & . &  .\\ \hline
 % $p_n$ & $d_{n,1} $ & . &  $d_{n,j}$ & . & $d_{n,n}$ & MaxD(n) & MinD(n) \\ \hline 
\end{tabular}
 \caption{Notations and their Explanations}
 \vspace{-2mm}
\label{table1}
\end{table}

%  (i) A set of {\em potential} renewable energy generation sites $\{S_i\}, 1 \leq i \leq n$\\
% (ii) A set of demand (load) points $\{L_j\}, 1 \leq i \leq m$\\
% (iii) Cost $SCost_i$ of building a site at $S_i$\\
% (iv) Maximum capacity of generation at site $S_i$: $SCap_i$ \\
% (v) Variation of generation at site $S_i$ over time interval 0 to $T$, where each time interval may correspond to a month.\\
% (vi) Variation of demand (load) at load point $L_j$ over time interval 0 to $T$ \\
% (vii) Budget allocated for design renewable energy infrastructure $B$\\

% \textbf{INPUT:}
% \begin{itemize}
%     \item Set of potential sites $\{S_i\}_{i=1}^n$
%     \item Set of demand points $\{L_j\}_{j=1}^m$
%     \item Cost $SCost_i$ of building a site at $S_i$
%     \item Maximum capacity of generation $SCap_i$ at $S_i$
%     \item Cost $LCost_{i,j}$ of building a transmission line connecting site $S_i$ to demand point $L_j$
%     \item Maximum capacity of the transmission line $LCap_{i,j}$ connecting site $S_i$ to demand point $L_j$
%     \item Variation of generation at site $S_i$ over time
%     \item Variation of demand at demand point $L_j$ over time
%     \item Budget of design B
    
% \end{itemize}

In order to meet the largest fraction of demand subject to budget and capacity constraints, we utilize an Integer Linear Program (ILP) to find the optimal solution to the problem. However, in order to utilize ILP to find a solution, we need to set up the problem in a manner that is conducive to ILP solution. In the following we describe our set up and provide an ILP to find the optimal solution to the problem.

% Given a set of potential sites $S_1, \ldots, S_n$ and a set of Load/Demand Points  $L_1, \ldots,L_m$. Two parameters $SCost_i$ and $SCap_i$ is associated with site $S_i$ indicating the cost of building a facility at $S_i$ and the maximum capacity of generation at $S_i$, if the facility is indeed built. We denote the potential transmission line from site $S_i$ to load point $L_j$ as $e_{i, j}$. Two parameters $LCost_{i, j}$ and $LCap_{i, j}$ is associated with transmission line $e_{i, j}$ indicating the cost of building a this transmission line and the maximum capacity of this transmission line, if the transmission line is indeed built. It may be noted that a facility at site $S_i$ may or may not be built and similarly a  transmission line from site $S_i$ to load point $L_j$ as may or may not be built. We use an indicator variable $x_i$ which only takes values 0 or 1 to indicate the decision whether or not a facility is built at site $S_i$. Specifically,

We introduce a binary indicator variable $x_i$ to indicate whether a facility is built at site $S_i$.
\vspace{-0.15in}
\[
    x_{i} = 
\begin{cases}
    1, \text{if a facility is built at site}~S_i \\
    0, \text{otherwise}
\end{cases}
\]
 Similarly, we use a binary indicator variable $y_{i, j}$ to indicate whether a transmission line $e_{i, j}$ (connecting generating site $S_i$ to the load point $L_j$) is built or not.

\[ y_{i, j}= 
\begin{cases}
    1,& \text{if transmission line} ~e_{i,j} \text{ is built}\\
    0,              & \text{otherwise}
\end{cases}
\]

% \begin{equation}
%     f^{v_i \rightarrow u_j}_{(v_k, v_l)} = 
% \begin{cases}
%     1, & \text{if flow from $v_i$ to $u_j$ uses the edge between} \\ 
%      & \text{nodes $(v_k, v_l)$, $v_k \in V_P, v_l \in V_P \cup V_Q$} \\
%     0,              & \text{otherwise}
% \end{cases}
% \end{equation}

Observe there is no point in building the transmission line $e_{i, j}$ if it is decided not to build a facility at the site $S_i$. Mathematically, it implies that if $x_i = 0$, then $y_{i, j}$ must also be 0. This condition may be ensured by the following constraint 
\[ y _{i, j} \leq x_i\]

The {\em Total Cost} of construction of facilities and transmission lines is given by \[ \sum_{i = 1}^n SCost_i  \cdot x_i + \sum_{i = 1}^n \sum_{j = 1}^m LCost_{i, j}  \cdot y_{i, j}\]

We require the construction cost of sites and transmission lines cannot exceed a budget $B$, and enforce this with the following constraint \[ \sum_{i = 1}^n SCost_i  \cdot x_i + \sum_{i = 1}^n \sum_{j = 1}^m LCost_{i, j}  \cdot y_{i, j} \leq B\]

In order to capture energy flow from the generating sites to the demand points, we create a bipartite graph where the nodes on one side of bi-partition correspond to the potential sites ($S_i$), and the other side of bi-partition corresponds to the demand points ($L_j$). We add a directed edge from every potential site $S_i$ to every demand point $L_j$, resulting in a {\em complete} bipartite graph. Each node corresponding to $S_i$ has a cost $SCost_i$ and a capacity value $SCap_i$ associated with it. Similarly, each link $e_{i, j}$ has two parameters, cost ($LCost_{i, j}$) and capacity ($LCap_{i, j}$) associated with it. Finally, we add two additional nodes, SuperSource ($SS$) and SuperLoad ($SL$), add a directed edge from $SS$ to every $S_i$, and direct edges from every $L_j$ to $SL$. The cost of these edges is taken to be zero. Let $Cap(SS \rightarrow S_i) = SCap_i  \cdot x_i$, as this capacity will be equal to 0 if the decision is made not to build a facility at st $S_i$, i.e., $x_i = 0$. Similarly, $Cap(S_{i} \rightarrow L_{j}) = LCap_{i,j} \cdot y_{i, j}$, as the capacity will be equal to 0 if the decision is made not to build the transmission line corresponding to $e_{i, j}$. The capacity of $Cap(L_j \rightarrow SL)$ link is set equal to the demand, $D_{L_j}$, at the demand point $L_j$, to ensure that energy supplied at $L_j$ doesn't exceed the demand at that load point.  With this formulation, if we compute the {\em maximum flow} from the SuperSource ($SS$) to the SuperLoad ($SL$), that provides the solution to the optimal site election problem. A graph with six potential sites and four demand points is shown in Figure \ref{flowdia1}.

\begin{figure}[tbh]
	\begin{center}
	\vspace{-4.0mm}
		\includegraphics[width = 0.5\textwidth, keepaspectratio]{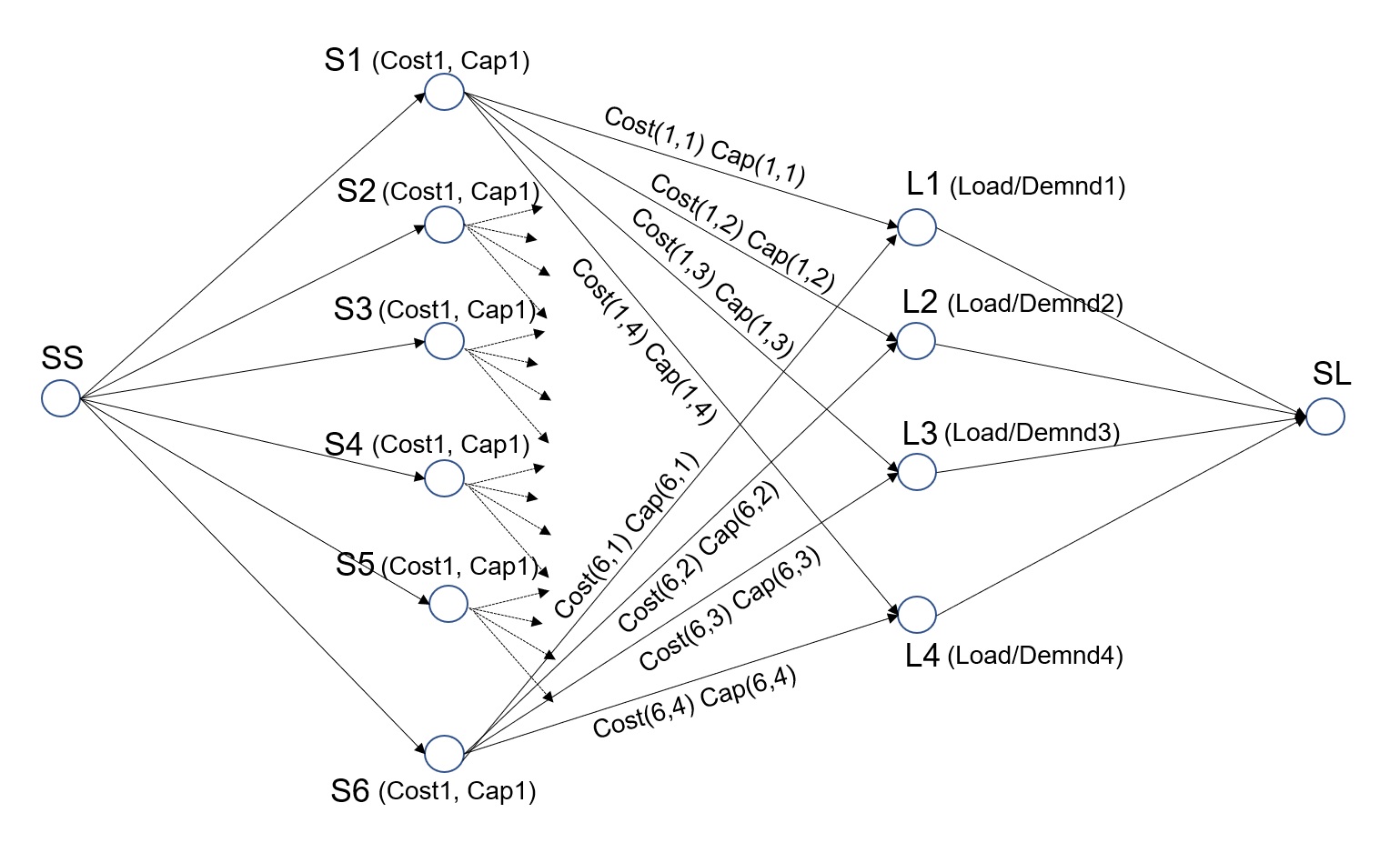}
	    \vspace{-4.0mm}
        \caption{Energy flow diagram from the potential sites to the load/demand points}
        %\vspace{-4.0mm}
		\label{flowdia1}      
	\end{center}
    \vspace{-5mm}
\end{figure}

If there are $n$ potential sites and $m$ load (demand) points, the corresponding network graph as shown in Figure \ref{flowdia1} will have $(n + m + 2)$ nodes and  $n + m + n \cdot m$ directed edges. The nodes will correspond to $SS$, $SL$, $S_{i}, 1 \leq i \leq n$ and $L_{j}, 1 \leq j \leq m$. There will be three groups of edges:\\

%\vspace{0.1in}
\noindent
(i) $SS \rightarrow S_{i}, 1 \leq i \leq n$,\\ (ii)  $S_{i} \rightarrow L_{j}, 1 \leq i \leq n, 1 \leq j \leq m$\\ (iii) $L_{j} \rightarrow SL, 1 \leq j \leq m$,\\

\noindent
Energy flow on a link $S_i \rightarrow L_j$ will be denoted as $f(i \rightarrow j)$. Network flow constraints can be written as follows:\\

\noindent
(i) $\forall_{i}$, $f(SS \rightarrow S_{i}) \leq Cap(SS \rightarrow S_i) =SCap_i \cdot x_i$

\noindent
(ii) \[\forall_{i,j}, \text{ }f(S_i \rightarrow L_j) \leq Cap(S_i \rightarrow L_j) = LCap_{i,j} \cdot y_{i.j}\]

\noindent
(iii) $\forall_{j}$, $f(L_j \rightarrow SL) \leq D_{L_j}$

\noindent
(iv) {\em Total Flow, F} at the node $SS$: \[F = \sum_{i = 1}^n f(SS \rightarrow S_{i}), 1 \leq i \leq n\]

\noindent
(v) At node $S_{i}, 1 \leq i \leq n$: \vspace{-3mm}\[f(SS \rightarrow S_{i}) = \sum_{j = 1}^m f(S_{i} \rightarrow L_{j})\]

\vspace{-3mm}
\noindent
(vi) At node $L_{j}, 1 \leq i \leq m$: \vspace{-3mm}\[f(L_{j}\rightarrow SL) = \sum_{i = 1}^n f(S_{i} \rightarrow L_{j})\]

\noindent
Objective: \textit{Maximize} $F$ (\textit{Total Flow})\\

It may be noted that the ILP formulation above will provide the solution only for one one-time sub-interval during which the generation at the sites $S_i$ and demand at load points $L_j$ remain constant. However, the formulation above can be easily extended to take care of multiple sub-intervals, where both supply and demand might change by replicating the network shown in Figure \ref{flowdia1} as many times as there are time sub-interval and also replicating the ILP constraints with changed values for supply and demand during each of the sub-intervals.  A modified graph with four-time sub-intervals is shown in Figure \ref{flowdia2}.

% The goal of this optimization problem is to {\em maximize F} subject to all network flow and link capacity constraints.

% \begin{figure}[tbh]
% 	\begin{center}
% 	%\vspace{-4.0mm}
% 		\includegraphics[width = 0.5\textwidth, keepaspectratio]{FlowDiagram_1.jpg}
% 	   % \vspace{-4.0mm}
%         \caption{Energy flow diagram from the potential sites to the load/demand points}
%         %\vspace{-4.0mm}
% 		\label{flowdia1}      
% 	\end{center}
% \end{figure}

\begin{figure}[tbh]
	\begin{center}
	\vspace{-4.0mm}
		\includegraphics[width = 0.35\textwidth, keepaspectratio]{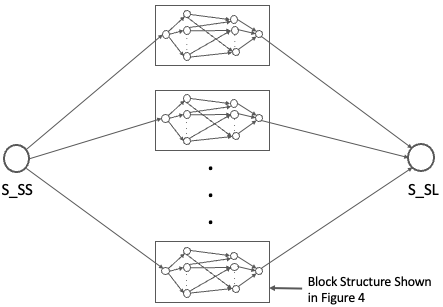}
	   % \vspace{-4.0mm}
        \caption{Energy flow diagram from the potential sites to the load/demand points}
        %\vspace{-4.0mm}
		\label{flowdia2}      
	\end{center}
    \vspace{-5mm}
\end{figure}

\section{Experimental Results}

In this section, we present the results of our experimental evaluation of the site selection process. In our evaluation, we studied the impact of \textit{(i) budget, (ii) increase in supply, and (iii) demand} on the \textit{percentage of demand met}. Our experimentation was conducted on synthetic data. However, \textit{we made a significant effort to make the synthetic data as realistic as possible}. In that vain, we collected publicly available domain data from most of the operating solar farms in Arizona. There are 15 operating solar farms in Arizona generating from 18 to 400MW of power. The data related to the cost of setting up a solar farm and associated infrastructure is sparsely available in the public domain. The cost of setting up farms varied from 19 to 401 million dollars. The cost of setting up 45 miles of 500kV transmission line is approximated at 95 million dollars. We created our synthetic data using the above figures as a baseline. We created three groups of datasets, each comprising 100 members. These three groups were created to study the impacts of (i) budget, (ii) supply, and (iii) demand. Each group was further divided into ten sub-groups, where each sub-group corresponds to a solar farm design problem instance, with variations in budget, supply, and demand. Figures \ref{budget}, \ref{sourceinc}, and \ref{deminc} show the impact of budget, supply, and demand, respectively, on the \textit{percentage of demand met}.

In our experimentation, we used 300 datasets and each dataset had 14 solar farms. In alignment with our goal of creating synthetic data as realistic as possible, the data values used in our experiments were derived from the generation capacity of existing solar farms in Arizona. The demand values and the site building budget numbers were also collected from the same sources. The generation capacity of these solar farms varied from 15 MW to 403 MW. Each dataset had 9 demand points, and demands in these points varied from 0 to 4057.48 MW. The cost of building a site facility varied from 13.63 to 402.14 million dollars. The cost and the capacity of the transmission lines connecting site $S_i$ to demand (load) point $L_j$ varied from 1.27 to 117.30 million dollars and from 250 to 1000 MW, respectively.

\begin{figure}[tbh]
	\begin{center}
	%\vspace{-4.0mm}
		\includegraphics[width = 0.48\textwidth, keepaspectratio]{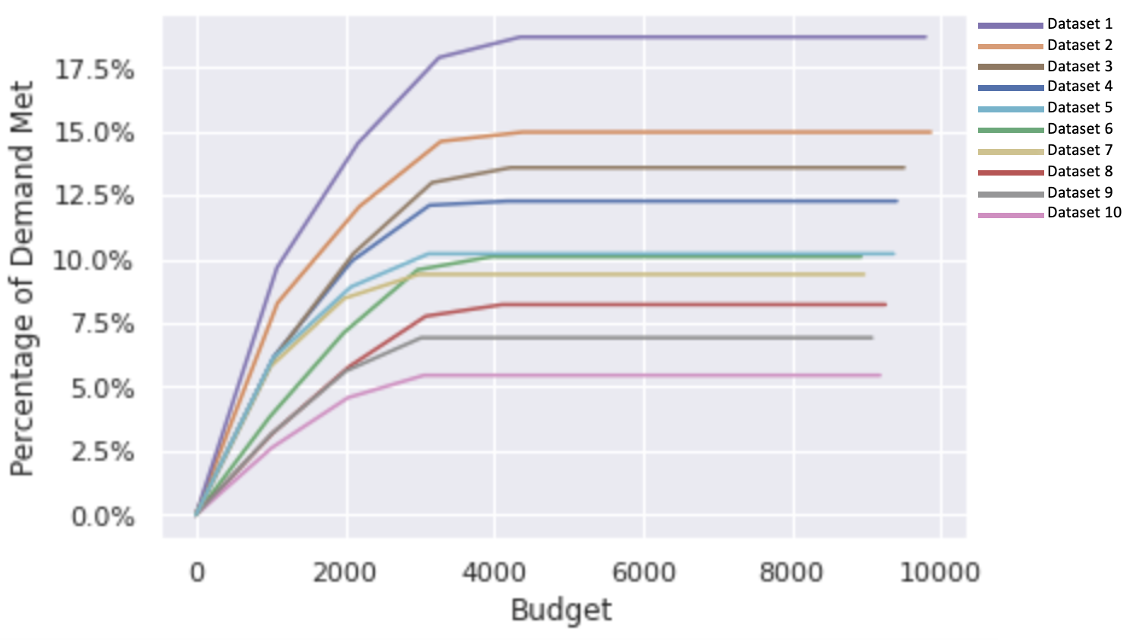}
	   % \vspace{-4.0mm}
        \caption{Impact of budget on the percentage of demand met}
        %\vspace{-4.0mm}
		\label{budget}      
	\end{center}
        \vspace{-5mm}
\end{figure}

\begin{figure}[tbh]
	\begin{center}
	%\vspace{-4.0mm}
		\includegraphics[width = 0.48\textwidth, keepaspectratio]{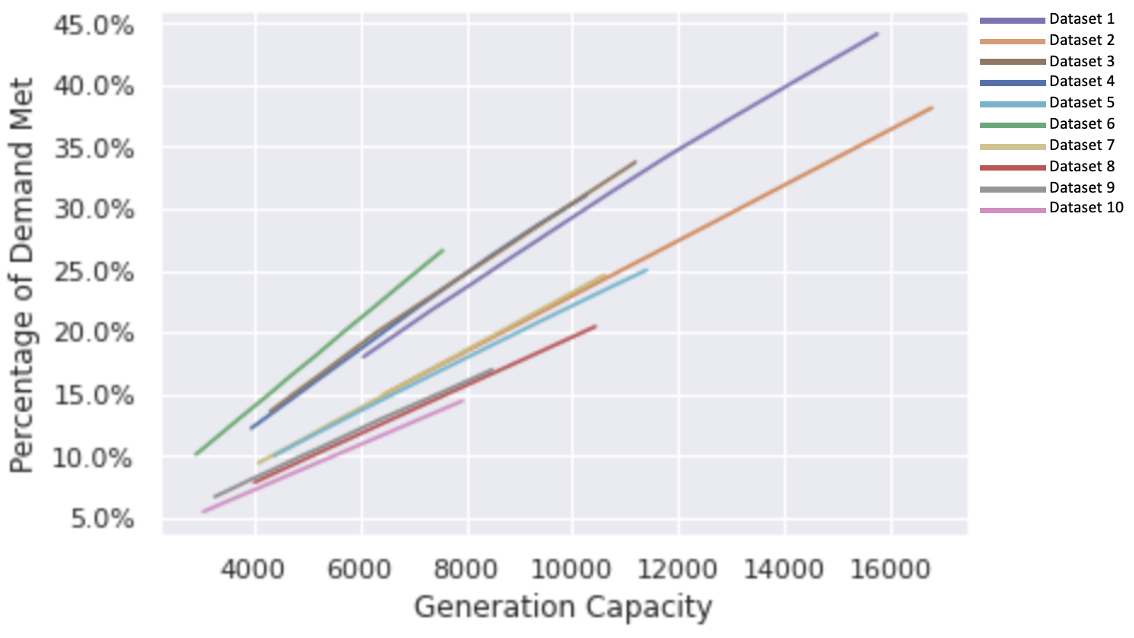}
	   % \vspace{-4.0mm}
        \caption{Impact of increase in the generation capacity at the sites}
        %\vspace{-4.0mm}
		\label{sourceinc}      
	\end{center}
        \vspace{-5mm}
\end{figure}

\begin{figure}[tbh]
	\begin{center}
	%\vspace{-4.0mm}
		\includegraphics[width = 0.48\textwidth, keepaspectratio]{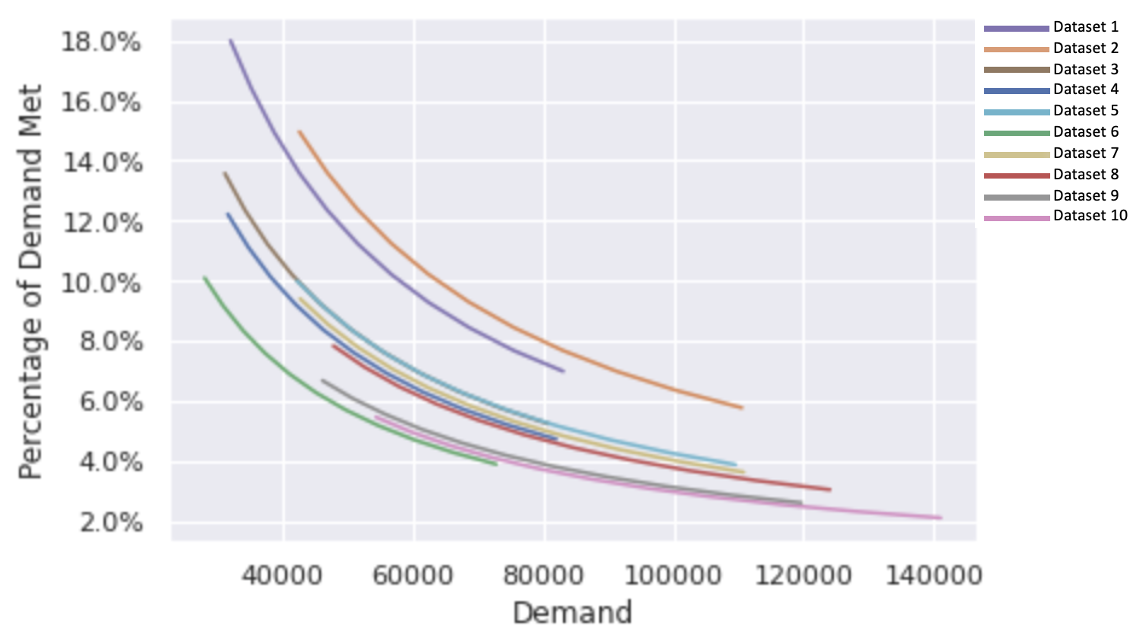}
	   % \vspace{-4.0mm}
        \caption{Impact of increase in the demand}
        %\vspace{-4.0mm}
		\label{deminc}      
	\end{center}
        \vspace{-5mm}
\end{figure}

From public domain information made available by operating solar farms in Arizona, we collected data related to (i) site establishment cost, (ii) maximum site capacity, and (iii) monthly generation of energy. From SRP, we collected data on (iv) transmission line capacity, (v) transmission line construction cost, and (vi) monthly variation of energy demand. To create our synthetic data, we utilized a maximum variance vector, $\big[\sigma_i\big] {1 \leq i \leq 6}$, for the six parameters indicated earlier. We created 300 datasets by tweaking the original dataset with the variance setup and including a range of budget values, generation capacities, and demand values for generating the plots in Figures \ref{budget}, \ref{sourceinc}, and \ref{deminc}. It may be noted that the difference between the two plots in these figures corresponding to datasets $i$ and $j$ is due to the difference in their corresponding variance parameters sampled from $\big[\sigma_i\big] ({1 \leq i \leq 6})$.

The 300 datasets were split into three sets of 100 each. The impact of variation of budget, on the first set is shown in Figure \ref{budget}. The impact of variation of generation capacities of the sites on the second set is shown in Figure \ref{sourceinc}. The impact of variation of demand on the third set is shown in Figure \ref{deminc}. Figure \ref{budget} shows 10 different plots where each curve corresponds to results obtained from ten different datasets. Similar settings apply for the curves in Figures \ref{sourceinc} and \ref{deminc}.

It may be observed that each curve in Figure \ref{budget} reaches a saturation point for a certain value of budget, and a further increase in budget has no impact on the percentage of demand met. Despite its counter-intuitive nature, this phenomenon occurs because each power station and transmission line has fixed generation and transmission capacities. Once all are activated, only a certain amount of energy can be relayed from the generation sites to the demand points. At that point in time, only a certain percentage of the total demand can be met. Further, increase in budget is not going to change that percentage.

In the experiments shown in Figure \ref{sourceinc}, the \textit{demand met} consistently rose with increased generation capacity, indicating that transmission line capacity wasn't a limiting factor. It may be observed that some of the curves stopped after a smaller value of generation capacity, whereas others continued for a larger value. This is due to the fact that the variation of the generation capacity of a site is limited by its maximum capacity, and the maximum capacity of some sites is much smaller than other sites. 

In all our experimentations presented in Figure \ref{deminc}, the percentage of demand met decreased with an increase in demand. This is quite intuitive, although we recognize that due to demands at different demand points being different, we could have seen somewhat different behavior. It just so happened that we did not observe any counter-intuitive phenomenon.

Our experiments were conducted on a platform running Ubuntu 22.04 LTS with an AMD Ryzen 9 4900HS CPU equipped with 16GB RAM. For code deployment, we used Python 3.10 with the Gurobi 10.0.1 library. The computation time for each dataset in our experiments never exceeded five minutes.

\section{Conclusion}
\vspace{-2mm}
In this paper, we presented techniques for the optimal selection of renewable energy sites that goes beyond conventional ways of selecting sites, utilizing GIS and AHP. We propose a two-phase approach where phase I, using GIS and AHP, does the preliminary selection, and phase II, utilizing mathematical programming, does the final selection. We conducted significant experimentation using synthetic data that attempts to mimic real data associated with currently operating solar farms in Arizona.

\printbibliography

\end{document}